\documentclass[journal,12pt,onecolumn,draftclsnofoot,]{IEEEtran}
\usepackage{geometry}
\usepackage{graphicx}
\usepackage{epstopdf}
\usepackage{amsmath}
\usepackage{amssymb}
\usepackage{amsfonts}
\usepackage{mathtools}
\newcommand{\e}[1]{{\mathbb E}\left[ #1 \right]}
\newcommand{\p}[1]{{\mathbb P}\left[ #1 \right]}
\newcommand{\defeq}{\vcentcolon=}
\usepackage{amsthm}

\DeclareSymbolFont{bbold}{U}{bbold}{m}{n}
\DeclareSymbolFontAlphabet{\mathbbold}{bbold}

\newtheorem{theorem}{Theorem}

\newtheorem{lemma}[theorem]{Lemma}

\newtheorem{prop}[theorem]{Property}
\newtheorem{remark}[theorem]{Remark}

\allowdisplaybreaks
\title{\mbox{Feedback-Capacity of Degraded\! Gaussian\! Vector\! BC} using Directed Information and Concave Envelopes}
\author{
\IEEEauthorblockN{Viswanathan~Ramachandran and Sibi~Raj~B~Pillai}\\
\IEEEauthorblockA{Department of Electrical Engineering\\
        Indian Institute of Technology Bombay }
}
\begin{document}

\maketitle

\noindent \begin{abstract}
\textbf{
It is known that the capacity region of a two user physically degraded discrete memoryless (DM) broadcast channel (BC) is not enlarged by feedback. An identical result holds true for a physically degraded Gaussian BC, established later using a variant of the Entropy Power Inequality (EPI). In this paper, we extend the latter result to a physically degraded Gaussian Vector BC (PD-GVBC). However, the extension is not EPI based, but employs a recent result on the factorization of concave envelopes. While the existing concave envelope factorization results do not hold in the presence of feedback, we show that factorizing the corresponding directed information quantities suffice to attain the feedback capacity region of a PD-GVBC. Our work demonstrates that factorizing concave envelopes of directed information can handle situations involving feedback. We further show that the capacity region of a discrete memoryless reversely physically degraded BC is not enlarged by feedback.}
\end{abstract}

\section{Introduction}

The capacity region of a  physically degraded (PD) discrete memoryless (DM) broadcast channel (BC) is not enlarged by the presence of perfect feedback from both the receivers, a result which was established by El Gamal~\cite{gamal1978feedback}. The same conclusion was arrived at for PD Gaussian BCs with feedback by El Gamal~\cite{gah4al1981capacity}. The latter proof involved a variant of the Entropy Power Inequality (EPI), in order to incorporate the presence of feedback.

\par{Evaluating the capacity regions or bounds to it in the additive Gaussian settings many-a-times involve the computation of extremal auxiliaries, for example, the broadcast channel~\cite{el2011network}. The standard technique for proving the optimality of Gaussian auxiliaries in additive Gaussian noise settings is the EPI, or its variants. The application of EPI may not always be directly viable, and can be cumbersome, see~\cite{weingarten2006capacity}. Recently, Geng and Nair~\cite{geng2012capacity}  proposed a technique for proving the optimality of Gaussian auxiliaries and used it to establish the capacity region of a two receiver Gaussian vector BC (GVBC) with both private and common messages. Their proof technique serves as a new tool for proving converses in Gaussian settings, instead of applying the traditional EPI.}

\par{In this paper, we extend the result in \cite{gah4al1981capacity} for PD Gaussian BCs to the vector BC case. This is done by extending the technique of concave envelopes~\cite{geng2012capacity}. However this extension is not very straightforward. In particular, \cite{geng2012capacity} requires the factorization of concave envelopes of multiletter mutual information terms. We show that this may not hold in the presence of feedback. Our major contribution is a reformulation of the concave envelopes in terms of \emph{directed information} instead of mutual information. Interestingly, this factorization helps in concluding that the capacity region of a PD-GVBC is not enlarged by feedback. The novelty of the paper also lies in bringing out the fact that factorizing concave envelopes of directed information can handle situations involving feedback, where concave envelopes of mutual information are not factorizable.}

\par{ We also show that feedback does not enlarge the capacity region of a reversely physically degraded DM broadcast channel (RPDBC), a result of independent interest. The paper is organized as follows. We introduce the system model in Section \ref{chm}. In Section \ref{conenv}, we review the notion of concave envelopes in Geng and Nair's context~\cite{geng2012capacity} and the notion of directed information. Section \ref{gbcfb} describes the feedback capacity region of a PD-GVBC. Section \ref{rddmbc} deals with the  feedback capacity region of a RPDBC. Finally Section \ref{concl} concludes the paper.}

\section{System Model}\label{chm}

Consider the broadcast model in Fig 1. Here, $M_{1}$ and $M_{2}$ are two independent 
messages to the respective receivers. $M_i$ is uniformly distributed over $[1:2^{nR_{i}}], \,i \texttt=1,2$. The BC $q(\mathbf{y},\mathbf{z}{\vert}\mathbf{x})$ is physically degraded i.e. $q(\mathbf{y},\mathbf{z}{\vert}\mathbf{x})=p_{1}(\mathbf{y}{\vert}\mathbf{x})p_{2}(\mathbf{z}{\vert}\mathbf{y})$. More specifically, $\textbf{Y}=\textbf{X}+\textbf{N}$ and $\textbf{Z}=\textbf{Y}+\tilde{\textbf{N}}$, where $\textbf{N}$ and $\tilde{\textbf{N}}$ are independent of each other, and of $\textbf{X}$. There is perfect feedback from both receivers to the encoder. 

\begin{figure}[h]
\begin{center}
\includegraphics[scale=1.2]{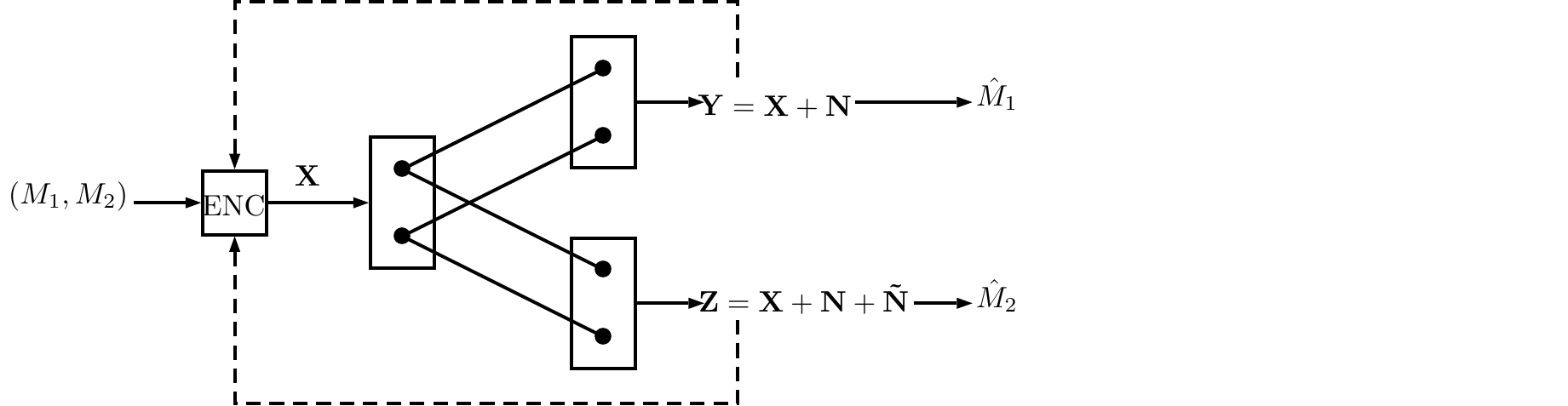}
\label{model}
\caption{Gaussian PD Vector BC with feedback} 
\end{center}
\end{figure}

We will assume without loss of generality that all the random vectors are zero mean. Here $\mathbf{N} \sim \mathcal{N}(\mathbf{0},\mathbf{K})$ and $\tilde{\mathbf{N}} \sim \mathcal{N}(\mathbf{0},\tilde{\mathbf{K}})$. The input is subject to a covariance matrix constraint $\e{{\mathbf{X}}{{\mathbf{X}}^{\mathbf{T}}}} {\preccurlyeq} {{\mathbf{K}}^{'}}$. A $(2^{nR_{1}},2^{nR_{2}},n,\lambda_n)$ code for the channel is defined to be a pair of encoder mappings, $\mathbf{X}_{i}=f_{i}(M_{1},M_{2},\mathbf{Y}^{i-1},\mathbf{Z}^{i-1}), \,i \texttt=1,2$ and two decoder mappings $d_{i} : {\mathbf{\mathbb{R}}}^{n} \to [1:2^{nR_{i}}], i=1,2$. The probability of error $\lambda_n$ is given by
\begin{equation}
P_{e,1}^{(n)}=\frac{1}{2^{nR_{1}}} \sum_{m_{1}} \p{(d_{1}(\mathbf{Y}^{n}){\neq}m_{1}{\vert}m_{1}\:\: \text{sent})}
\end{equation}
\begin{equation}
P_{e,2}^{(n)}=\frac{1}{2^{nR_{2}}} \sum_{m_{2}} \p{(d_{2}(\mathbf{Z}^{n}){\neq}m_{2}{\vert}m_{2}\:\: \text{sent})}
\end{equation}
\begin{equation}
\lambda_{n}=\max{(P_{e,1},P_{e,2})}.
\end{equation}
A rate pair $(R_{1},R_{2})$ is said to be achievable if there exists a sequence of $(2^{nR_{1}},2^{nR_{2}},n,\lambda_n)$ codes with $\lambda_{n} \to 0$ as $n \to \infty$. The capacity region $\mathcal{C}_{VBC}^{fb}$ is defined to be the closure of the set of achievable $(R_{1},R_{2})$ pairs.

El Gamal~\cite{gamal1978feedback} showed that the capacity region of a physically degraded DM broadcast channel is not enlarged by feedback. Thus the capacity is given by superposition coding:
\begin{theorem}~\cite{gamal1978feedback}\label{thm:sup}
The capacity region of a DM physically degraded BC $q(y,z{\vert}x)$ with perfect feedback from both decoders is the set of $(R_{1},R_{2})$ pairs such that
\begin{gather}
R_{1} \leq I(X;Y{\vert}U)\\
R_{2} \leq I(U;Z),
\end{gather}
\noindent{for some $p(u,x)$, where $U$ is an auxiliary random variable with $U \to X \to Y \to Z$, and $|\mathcal{U}|{\leq}\min{{(|\mathcal{X}|,|\mathcal{Y}|,|\mathcal{Z}|)}}+1$.}
\end{theorem}
\noindent El Gamal~\cite{gah4al1981capacity} further showed that even in the scalar Gaussian case, feedback does not improve the capacity region if the channel is physically degraded. This result required an adaptation of the EPI.

\par{As seen in Theorem \ref{thm:sup}, the capacity regions or bounds thereof are often represented in terms of auxiliary random variables. The evaluation of these regions reduce to optimization problems over the joint distribution of input and auxiliary random variables. In Gaussian problems, proving the optimality of Gaussian auxiliaries generally involves an application of the EPI, or its variants~\cite{el2011network}. The application of EPI can be cumbersome in some cases. For instance, the converse for the Gaussian Multiple Input Multiple Output (MIMO) BC problem required the introduction of the notion of an enhanced channel~\cite{weingarten2006capacity} in order to make the application of EPI viable. However, even the result of \cite{weingarten2006capacity} cannot handle a vector Gaussian BC with both private and common messages. The latter problem was recently solved by Geng and Nair~\cite{geng2012capacity} by the method of factorization of concave envelopes. But the concave envelopes in \cite{geng2012capacity} do not factorize in the presence of feedback. We reformulate the factorization of concave envelopes in terms of directed information to yield the capacity region of a PD-GVBC with feedback. Our main result is the following theorem:}
\begin{theorem} \label{thm:pdvecfb}
The capacity region $\mathcal{C}_{VBC}^{fb}$ of a PD-GVBC with perfect feedback from both the receivers to the transmitter is the union of the set of $(R_{1},R_{2})$ pairs such that
\begin{gather}
R_{1} \leq \frac{1}{2} \log\Bigg(\frac{\left\vert{\textbf{B}_{\textbf{1}}+\textbf{K}}\right\vert}{\left\vert{\textbf{K}}\right\vert}\Bigg) \\
R_{2} \leq \frac{1}{2} \log\Bigg(\frac{\left\vert{ \textbf{B}_{\textbf{2}}+\textbf{B}_{\textbf{1}}+\textbf{K}+\tilde{\textbf{K}}}\right\vert} {\left\vert{ \textbf{B}_{\textbf{1}}+\textbf{K}+\tilde{\textbf{K}}}\right\vert}\Bigg),
\end{gather}
for some $\textbf{B}_{\textbf{1}},\textbf{B}_{\textbf{2}}\succcurlyeq0$, with $\textbf{B}_{\textbf{1}}+\textbf{B}_{\textbf{2}}\:{\preccurlyeq}\:\textbf{K}^{'}$.
\end{theorem}
\noindent Note that the region mentioned above is exactly the region without feedback.

\section{Concave Envelopes and Directed Information}\label{conenv}
Let us first review the technique of concave envelopes introduced in \cite{geng2012capacity}. We follow the notation used in that work. Recall that the upper concave envelope of a function $f(x)$ is the smallest concave function $g(x)$ such that $g(x) \geq f(x)$ throughout the domain of $f(x)$. Equivalently, $g(x)$ can be expressed as~\cite{geng2012capacity}
 \begin{equation} \label{conen}
g(x) = \sup_{p(x): \e{X}=x} \e{f(X)}.
\end{equation} 
For a BC $q(\mathbf{y},\mathbf{z}{\vert}\mathbf{x})$, and $\lambda > 1$, define the following function of $p(\mathbf{x})$ (the input distribution)
\begin{equation}
s_{\lambda}^{q}(\mathbf{X}) \defeq I(\mathbf{X};\mathbf{Y}) - {\lambda}I(\mathbf{X};\mathbf{Z}).
\end{equation}
For $(V,\mathbf{X})$  such that $V \to \mathbf{X} \to (\mathbf{Y},\mathbf{Z})$, define
\begin{equation}
s_{\lambda}^{q}(\mathbf{X}{\vert}V) \defeq I(\mathbf{X};\mathbf{Y}{\vert}V) - {\lambda}I(\mathbf{X};\mathbf{Z}{\vert}V).
\end{equation}
Now we denote the upper concave envelope of $s_{\lambda}^{q}(\mathbf{X})$ as
\begin{equation*}
S_{\lambda}^{q}(\mathbf{X}) \defeq \mathfrak{C}(s_{\lambda}^{q}(\mathbf{X})).
\end{equation*}
Applying the definition from Equation \eqref{conen}, we get
\begin{align}
\mathfrak{C}(s_{\lambda}^{q}(\mathbf{X})) &= \sup_{\substack{p(v{\vert}\mathbf{x}):\\ V \to \mathbf{X} \to (\mathbf{Y},\mathbf{Z})}} s_{\lambda}^{q}(\mathbf{X}{\vert}V)\\
&= \sup_{\substack{p(v{\vert}\mathbf{x}):\\ V \to \mathbf{X} \to (\mathbf{Y},\mathbf{Z})}} I(\mathbf{X};\mathbf{Y}{\vert}V) - {\lambda}I(\mathbf{X};\mathbf{Z}{\vert}V).
\end{align}
Define a conditional version of the concave envelope
\begin{equation}
S_{\lambda}^{q}(\mathbf{X}{\vert}W) \defeq \sum_{w} p(w)S_{\lambda}^{q}(\mathbf{X}{\vert}W=w).
\end{equation}
Since $S_{\lambda}^{q}(\mathbf{X})$ is concave in $p(\mathbf{x})$, by Jensen's inequality
\begin{equation}
S_{\lambda}^{q}(\mathbf{X}{\vert}W) \leq S_{\lambda}^{q}(\mathbf{X}).
\end{equation}
For a two-letter broadcast channel $q(\mathbf{y}_{1},\mathbf{z}_{1}{\vert}\mathbf{x}_{1}){\times}q(\mathbf{y}_{2},\mathbf{z}_{2}{\vert}\mathbf{x}_{2})$, define similarly a function of $p(\mathbf{x}_{1},\mathbf{x}_{2})$
\begin{equation*}
s_{\lambda}^{q{\times}q}(\mathbf{X}_{1},\mathbf{X}_{2}) \defeq I(\mathbf{X}_{1},\mathbf{X}_{2};\mathbf{Y}_{1},\mathbf{Y}_{2}) - {\lambda}I(\mathbf{X}_{1},\mathbf{X}_{2};\mathbf{Z}_{1},\mathbf{Z}_{2}).
\end{equation*}
Similar definitions apply for the quantities $s_{\lambda}^{q{\times}q}(\mathbf{X}_{1},\mathbf{X}_{2}{\vert}V)$, $S_{\lambda}^{q{\times}q}(\mathbf{X}_{1},\mathbf{X}_{2})$ and $S_{\lambda}^{q{\times}q}(\mathbf{X}_{1},\mathbf{X}_{2}{\vert}W)$.

\par{The key technique of \cite{geng2012capacity} was to show that the concave envelope defined over the joint distribution $p(\mathbf{x}_{1},\mathbf{x}_{2})$ of a product broadcast channel $q(\mathbf{y}_{1},\mathbf{z}_{1}{\vert}\mathbf{x}_{1}){\times}q(\mathbf{y}_{2},\mathbf{z}_{2}{\vert}\mathbf{x}_{2})$ satisfies the following subadditivity property:
\begin{gather}\label{geng}
S_{\lambda}^{q{\times}q}(\mathbf{X}_{1},\mathbf{X}_{2}) \leq S_{\lambda}^{q}(\mathbf{X}_{1})+S_{\lambda}^{q}(\mathbf{X}_{2}).
\end{gather}
This in turn leads to the optimality of Gaussian auxiliary random variables in establishing the capacity region of a GVBC.}

\par{However, the subadditivity mentioned above may not hold in the presence of feedback. This can be seen by considering an example of a single user channel from $\textbf{X} \to \textbf{Z}$. Consider a two letter channel $p(\textbf{z}_{1}{\vert}\textbf{x}_{1}) \times p(\textbf{z}_{2}{\vert}\textbf{x}_{2})$. Notice that, due to feedback, though $\textbf{X}_{1} \to \textbf{X}_{2} \to \textbf{Z}_{2}$ holds, $\textbf{Z}_{1} \to \textbf{X}_{1} \to \textbf{X}_{2}$ may not. Hence the term $I(\textbf{X}_{2};\textbf{Z}_{1}{\vert}\textbf{X}_{1})$ may not be zero and consequently, the subadditivity $I(\textbf{X}_{1},\textbf{X}_{2};\textbf{Z}_{1},\textbf{Z}_{2}) \leq I(\textbf{X}_{1};\textbf{Z}_{1})+I(\textbf{X}_{2};\textbf{Z}_{2})$ is not guaranteed to hold. Thus the technique of \cite{geng2012capacity} as such does not extend to the case with feedback.}\\
\par{The main contribution of the current work is the formulation of concave envelopes in terms of directed information, which admit a convenient factorization/subadditivity property over product channels. This in turn is used to prove the optimality of Gaussian auxiliary random variables, establishing the capacity region of a PD-GVBC with feedback. Recall that the notion of directed information between two sequences, introduced by Massey is~\cite{massey1990causality} 
\begin{equation}
I(X^{N} \to Y^{N})=\sum_{n=1}^{N} I(X^{n};Y_{n}{\vert}Y^{n-1}).
\end{equation}
Note that $I(X^{N} \to Y^{N}) \neq I(Y^{N} \to X^{N})$ in general. The following lemma relates directed information and mutual information.
\begin{lemma}~\cite{massey1990causality}\label{massey}
If $X^{N}$ is the input and $Y^{N}$ is the output of a channel, then
\begin{equation}
I(X^{N} \to Y^{N}) \leq I(X^{N};Y^{N})
\end{equation}
with equality if and only if the channel is memoryless and used without feedback.
\end{lemma}}

\section{PD-GVBC with feedback}\label{gbcfb}
In this Section, we denote $s_{\lambda}^{q}(\mathbf{X})$ by $\overrightarrow{s}_{\lambda}^{q}(\mathbf{X})$ and $S_{\lambda}^{q}(\mathbf{X})$ by $\overrightarrow{S}_{\lambda}^{q}(\mathbf{X})$ to emphasize the fact that we are now replacing mutual information by directed information. Similarly, for product channels,
\begin{align}
\overrightarrow{s}_{\lambda}^{q{\times}q}(\mathbf{X}_{1},\mathbf{X}_{2}) &\defeq I(\mathbf{X}_{1},\mathbf{X}_{2} \to \mathbf{Y}_{1},\mathbf{Y}_{2})- {\lambda}I(\mathbf{X}_{1},\mathbf{X}_{2} \to \mathbf{Z}_{1},\mathbf{Z}_{2}).
\end{align}


\def\OKp{\mathcal{O}_{{{\mathbf{K}}^{'}}}}

\noindent Let the covariance matrix constraint on the input be 
$\e{{\mathbf{X}}{{\mathbf{X}}^{\mathbf{T}}}} {\preccurlyeq} {{\mathbf{K}}^{'}}$.
Let us call the region in Theorem \ref{thm:pdvecfb} as $\mathcal{I}_{{{\mathbf{K}}^{'}}}$. 
Clearly $\mathcal{I}_{{{\mathbf{K}}^{'}}}$ can be achieved by ignoring the feedback and 
using superposition coding. On the other hand, Theorem~\ref{thm:sup} suggests that any achievable
rate pair $(R_1,R_2) \in \OKp$, where 
\begin{equation}\label{okp}
\OKp =  \bigcup \left(I(V;\mathbf{Z}), I(\mathbf{X};\mathbf{Y}|V)\right).
\end{equation}  
Here the union is taken over all $p(v,\mathbf{x})$ such that $p(\mathbf{x})$ meets the covariance constraint 
${{\mathbf{K}}^{'}}$, and  $V \to \mathbf{X} \to (\mathbf{Y},\mathbf{Z})$.
We now focus on showing that ${\mathcal{O}}_{\textbf{K}^{'}} \subseteq {\mathcal{I}}_{\textbf{K}^{'}}$. 
Since  both $\mathcal{I}_{{{\mathbf{K}}^{'}}}$ and $\mathcal{O}_{{{\mathbf{K}}^{'}}}$ are closed convex sets, they can be expressed as an intersection of supporting hyperplanes. Recall that any closed convex set $\mathcal{S}$ belonging to the positive quadrant can be expressed as
\begin{gather}\label{supph1} 
\mathcal{S}=\cap_{{\lambda} \geq 1}\{(R_{1},R_{2}) \in {\mathbb{R}}_{+}^{2}: R_{1}+{\lambda}R_{2} \leq V_{\lambda}\},
\end{gather}
\noindent{where}
\begin{gather}\label{supph2}
V_{\lambda}=\max_{(R_{1},R_{2}) \in \mathcal{S}} R_{1}+{\lambda}R_{2}.
\end{gather}
\noindent Using the characterization in equations \eqref{supph1} and \eqref{supph2}, it suffices to show that
\begin{equation}\label{suphy}
\max_{(R_{1},R_{2}) \in {\mathcal{O}}_{{\textbf{K}^{'}}}} R_{1}+{\lambda}R_{2} \leq \max_{(R_{1},R_{2}) \in {\mathcal{I}}_{{\textbf{K}^{'}}}} R_{1}+{\lambda}R_{2}.
\end{equation}
From the LHS of \eqref{suphy}
\begin{align}
&\max_{(R_{1},R_{2}) \in {\mathcal{O}}_{{\textbf{K}}^{'}}} R_{1}+{\lambda}R_{2} \notag\\
&\leq  \sup_{\textbf{X}: \e{{\textbf{X}}{{\textbf{X}}^{\textbf{T}}}} {\preccurlyeq} {{\textbf{K}}^{'}}} {\lambda}I(V;\textbf{Z})+I(\textbf{X};\textbf{Y}_1|V) \notag\\
&=  \sup_{\textbf{X}: \e{{\textbf{X}}{{\textbf{X}}^{\textbf{T}}}} {\preccurlyeq} {{\textbf{K}}^{'}}} {\lambda}I(X;\textbf{Z})+I(\textbf{X};\textbf{Y}|V)-{\lambda}I(\textbf{X};\textbf{Z}|V) \notag\\
&\leq \sup_{\textbf{X}: \e{{\textbf{X}}{{\textbf{X}}^{\textbf{T}}}} {\preccurlyeq} {{\textbf{K}}^{'}}} {\lambda}I(\textbf{X};\textbf{Z}) + \sup_{\substack{V \to \textbf{X} \to (\textbf{Y},\textbf{Z}) \\ \e{{\textbf{X}}{{\textbf{X}}^{\textbf{T}}}} {\preccurlyeq} {{\textbf{K}}^{'}}}} \overrightarrow{s}_{\lambda}^{q}({\textbf{X}}|V),
\end{align}
where the first inequality follows from the outer bound in expression \eqref{okp} and the last step follows by the definition of $\overrightarrow{s}_{\lambda}^{q}({\textbf{X}}|V)$.
We can now prove Theorem \ref{thm:pdvecfb} by establishing that
\begin{align} \label{equi}
&\sup_{\textbf{X}: \e{{\textbf{X}}{\textbf{X}^{\textbf{T}}}} {\preccurlyeq} {\textbf{K}^{'}}} {\lambda}I(\textbf{X};\textbf{Z}) + \sup_{\substack{V \to \textbf{X} \to (\textbf{Y},\textbf{Z}) \\ {\e{{\textbf{X}}{\textbf{X}^{\textbf{T}}}} {\preccurlyeq} {\textbf{K}^{'}}}}} \overrightarrow{s}_{\lambda}^{q}({\textbf{X}}{\vert}V) \notag\\
& \leq \frac{1}{2} \log\Bigg(\frac{\left\vert{\textbf{B}_{\textbf{1}}+\textbf{K}}\right\vert}{\left\vert{\textbf{K}}\right\vert}\Bigg)+\frac{{\lambda}}{2} \log\Bigg(\frac{\left\vert{\textbf{B}_{\textbf{2}}+\textbf{B}_{\textbf{1}}+\textbf{K}+\widetilde{\textbf{K}}}\right\vert} {\left\vert{\textbf{B}_{\textbf{1}}+\textbf{K}+\widetilde{\textbf{K}}}\right\vert}\Bigg).
\end{align}
\par{We now propose some factorization properties that hold for product DM channels with feedback.}
\begin{prop} \label{factlem}
Consider a DM channel $q(\textbf{y}{\vert}\textbf{x})$ and its two letter extension $q(\textbf{y}_{\textbf{1}}{\vert}\textbf{x}_{\textbf{1}}) \times q(\textbf{y}_{\textbf{2}}{\vert}\textbf{x}_{\textbf{2}})$. Then the following holds
\begin{equation}\label{dmf}
I(\textbf{X}_{\textbf{1}},\textbf{X}_{\textbf{2}} \to \textbf{Y}_{\textbf{1}},\textbf{Y}_{\textbf{2}}) \leq I(\textbf{X}_{\textbf{1}} \to \textbf{Y}_{\textbf{1}})+I(\textbf{X}_{\textbf{2}} \to \textbf{Y}_{\textbf{2}})
\end{equation}
Further, equality is attained when $\textbf{Y}_{\textbf{1}}$ and $\textbf{Y}_{\textbf{2}}$ are independent.
\end{prop}
\begin{proof}
\begin{align}
&I(\textbf{X}_{\textbf{1}},\textbf{X}_{\textbf{2}} \to \textbf{Y}_{\textbf{1}},\textbf{Y}_{\textbf{2}})-I(\textbf{X}_{\textbf{1}};\textbf{Y}_{\textbf{1}})-I(\textbf{X}_{\textbf{2}};\textbf{Y}_{\textbf{2}}) \notag\\
&=I(\textbf{X}_{\textbf{1}};\textbf{Y}_{\textbf{1}})+I(\textbf{X}_{\textbf{1}},\textbf{X}_{\textbf{2}};\textbf{Y}_{\textbf{2}}{\vert}\textbf{Y}_{\textbf{1}})-I(\textbf{X}_{\textbf{1}};\textbf{Y}_{\textbf{1}})-I(\textbf{X}_{\textbf{2}};\textbf{Y}_{\textbf{2}}) \notag\\
&= I(\textbf{X}_{\textbf{1}},\textbf{X}_{\textbf{2}};\textbf{Y}_{\textbf{2}}{\vert}\textbf{Y}_{\textbf{1}})-I(\textbf{X}_{\textbf{2}};\textbf{Y}_{\textbf{2}}) \notag\\
&= h(\textbf{Y}_{\textbf{2}}{\vert}\textbf{Y}_{\textbf{1}})-h(\textbf{Y}_{\textbf{2}}{\vert}\textbf{Y}_{\textbf{1}},\textbf{X}_{\textbf{1}},\textbf{X}_{\textbf{2}})-h(\textbf{Y}_{\textbf{2}})+h(\textbf{Y}_{\textbf{2}}{\vert}\textbf{X}_{\textbf{2}}) \notag\\
&\leq h(\textbf{Y}_{\textbf{2}})-h(\textbf{N}_{\textbf{2}}{\vert}\textbf{Y}_{\textbf{1}},\textbf{X}_{\textbf{1}},\textbf{X}_{\textbf{2}})-h(\textbf{Y}_{\textbf{2}})+h(\textbf{N}_{\textbf{2}}{\vert}\textbf{X}_{\textbf{2}})= 0 \notag .
\end{align}
The equality is attained when $h(\textbf{Y}_{\textbf{2}}{\vert}\textbf{Y}_{\textbf{1}})=h(\textbf{Y}_{\textbf{2}})$.
\end{proof}
\noindent Thus the factorization inequality holds true for a single user channel with feedback. Now we show the factorization of concave envelopes for a two letter broadcast channel. The following lemma holds for such product broadcast channels.
\begin{prop} \label{lem:2letfact}
For a two letter broadcast channel $q(\textbf{y}_{\textbf{1}},\textbf{z}_{\textbf{1}}{\vert}\textbf{x}_{\textbf{1}}){\times}q(\textbf{y}_{\textbf{2}},\textbf{z}_{\textbf{2}}{\vert}\textbf{x}_{\textbf{2}})$,
\begin{align}\label{2letfact}
\overrightarrow{S}_{\lambda}^{q{\times}q}(\textbf{X}_{\textbf{1}},\textbf{X}_{\textbf{2}}) \leq \overrightarrow{S}_{\lambda}^{q}(\textbf{X}_{\textbf{1}})+\overrightarrow{S}_{\lambda}^{q}(\textbf{X}_{\textbf{2}})
\end{align}
\end{prop}
\begin{proof}
For any $(V,\textbf{X}_{\textbf{1}},\textbf{X}_{\textbf{2}})$ such that $V \to (\textbf{X}_{\textbf{1}},\textbf{X}_{\textbf{2}}) \to (\textbf{Y}_{\textbf{1}},\textbf{Y}_{\textbf{2}},\textbf{Z}_{\textbf{1}},\textbf{Z}_{\textbf{2}})$, we have
\begin{align}
&\overrightarrow{s}_{\lambda}^{q{\times}q}(\textbf{X}_{\textbf{1}},\textbf{X}_{\textbf{2}}{\vert}V) \notag\\
&=I(\textbf{X}_{\textbf{1}},\textbf{X}_{\textbf{2}} \to \textbf{Y}_{\textbf{1}},\textbf{Y}_{\textbf{2}}{\vert}V) - {\lambda}I(\textbf{X}_{\textbf{1}},\textbf{X}_{\textbf{2}} \to \textbf{Z}_{\textbf{1}},\textbf{Z}_{\textbf{2}}{\vert}V) \notag\\
&= I(\textbf{X}_{\textbf{1}};\textbf{Y}_{\textbf{1}}{\vert}V)+I(\textbf{X}_{\textbf{1}},\textbf{X}_{\textbf{2}};\textbf{Y}_{\textbf{2}}{\vert}V,\textbf{Y}_{\textbf{1}})-{\lambda}I(\textbf{X}_{\textbf{1}};\textbf{Z}_{\textbf{1}}{\vert}V) -{\lambda}I(\textbf{X}_{\textbf{1}},\textbf{X}_{\textbf{2}};\textbf{Z}_{\textbf{2}}{\vert}V,\textbf{Z}_{\textbf{1}}) \notag\\
&\stackrel{(a)} {\leq} \overrightarrow{S}_{\lambda}^{q}(\textbf{X}_{\textbf{1}})+I(\textbf{X}_{\textbf{1}};\textbf{Y}_{\textbf{2}}{\vert}V,\textbf{Y}_{\textbf{1}},\textbf{X}_{\textbf{2}})+I(\textbf{X}_{\textbf{2}};\textbf{Y}_{\textbf{2}}{\vert}V,\textbf{Y}_{\textbf{1}})-{\lambda}I(\textbf{X}_{\textbf{1}};\textbf{Z}_{\textbf{2}}{\vert}V,\textbf{Z}_{\textbf{1}},\textbf{X}_{\textbf{2}})-{\lambda}I(\textbf{X}_{\textbf{2}};\textbf{Z}_{\textbf{2}}{\vert}V,\textbf{Z}_{\textbf{1}}) \notag\\
&\stackrel{(b)}= \overrightarrow{S}_{\lambda}^{q}(\textbf{X}_{\textbf{1}})+I(\textbf{X}_{\textbf{2}};\textbf{Y}_{\textbf{2}}{\vert}V,\textbf{Y}_{\textbf{1}})-{\lambda}I(\textbf{X}_{\textbf{2}};\textbf{Z}_{\textbf{2}}{\vert}V,\textbf{Z}_{\textbf{1}}) \notag\\
&\stackrel{(c)}{\leq} \overrightarrow{S}_{\lambda}^{q}(\textbf{X}_{\textbf{1}})+I(\textbf{X}_{\textbf{2}};\textbf{Y}_{\textbf{2}}{\vert}V,\textbf{Z}_{\textbf{1}})-{\lambda}I(\textbf{X}_{\textbf{2}};\textbf{Z}_{\textbf{2}}{\vert}V,\textbf{Z}_{\textbf{1}}) \notag\\
&\stackrel{(d)}{\leq} \overrightarrow{S}_{\lambda}^{q}(\textbf{X}_{\textbf{1}})+\overrightarrow{S}_{\lambda}^{q}(\textbf{X}_{\textbf{2}}{\vert}\textbf{Z}_{\textbf{1}}) \stackrel{(e)}{\leq} \overrightarrow{S}_{\lambda}^{q}(\textbf{X}_{\textbf{1}})+\overrightarrow{S}_{\lambda}^{q}(\textbf{X}_{\textbf{2}}) \notag,
\end{align}
where (a) and (d) follow from the definition of concave envelope, (b) follows from the Markov chain $\textbf{X}_{\textbf{1}} \to \textbf{X}_{\textbf{2}} \to (\textbf{Y}_{\textbf{2}},\textbf{Z}_{\textbf{2}})$, (c) follows from the fact that $\textbf{Z}_{\textbf{1}}$ is a degraded version of $\textbf{Y}_{\textbf{1}}$ and (e) follows from concavity. Now optimizing over all distributions $p(v|\mathbf{x}_1,\mathbf{x}_2)$ completes the proof.
\end{proof}
\noindent Thus we have shown that the concave envelopes in terms of directed information satisfy the factorization inequality \eqref{dmf}, \eqref{2letfact}. This in turn will be used to prove the optimality of Gaussian auxiliaries in achieving the capacity region of a PD-GVBC with feedback.\\
\par{We also require the following property regarding unitary transformations on a two-letter Gaussian channel.}
\begin{prop} \label{lem:rot}
Consider a two letter Gaussian channel with feedback
\begin{gather}
\textbf{Y}_{\textbf{1}}=\textbf{X}_{\textbf{1}}+\textbf{N}_{\textbf{1}}\\
\textbf{Y}_{\textbf{2}}=\textbf{X}_{\textbf{2}}+\textbf{N}_{\textbf{2}},
\end{gather}
\noindent{where $\textbf{N}_{\textbf{1}}$ and $\textbf{N}_{\textbf{2}}$ are independent and distributed as $\mathcal{N}(\textbf{0},{{\textbf{K}}^{'}})$. Define}
\begin{gather}
\textbf{U}_{\textbf{1}}=\frac{1}{\sqrt{2}}(\textbf{X}_{\textbf{1}}+\textbf{X}_{\textbf{2}}), \textbf{U}_{\textbf{2}}=\frac{1}{\sqrt{2}}(\textbf{X}_{\textbf{1}}-\textbf{X}_{\textbf{2}})\\
\textbf{V}_{\textbf{1}}=\frac{1}{\sqrt{2}}(\textbf{Y}_{\textbf{1}}+\textbf{Y}_{\textbf{2}}), \textbf{V}_{\textbf{2}}=\frac{1}{\sqrt{2}}(\textbf{Y}_{\textbf{1}}-\textbf{Y}_{\textbf{2}})
\end{gather}
Then $I(\textbf{X}_{\textbf{1}},\textbf{X}_{\textbf{2}} \to \textbf{Y}_{\textbf{1}},\textbf{Y}_{\textbf{2}}) = I(\textbf{U}_{\textbf{1}},\textbf{U}_{\textbf{2}} \to \textbf{V}_{\textbf{1}},\textbf{V}_{\textbf{2}})$.
\end{prop}
\begin{proof}
Take $\textbf{W}_{\textbf{1}} \texttt{=}\frac{1}{\sqrt{2}}(\textbf{N}_{\textbf{1}}+\textbf{N}_{\textbf{2}})$ and $\textbf{W}_{\textbf{2}}\texttt{=}\frac{1}{\sqrt{2}}(\textbf{N}_{\textbf{1}}-\textbf{N}_{\textbf{2}})$.
\begin{align}
&I(\textbf{X}_{\textbf{1}},\textbf{X}_{\textbf{2}} \to \textbf{Y}_{\textbf{1}},\textbf{Y}_{\textbf{2}}) - I(\textbf{U}_{\textbf{1}},\textbf{U}_{\textbf{2}} \to \textbf{V}_{\textbf{1}},\textbf{V}_{\textbf{2}}) \notag\\
&= h({\textbf{Y}}_{\textbf{1}},{\textbf{Y}}_{\textbf{2}})-h({\textbf{Y}}_{\textbf{1}}{\vert}{\textbf{X}}_{\textbf{1}})-h({\textbf{Y}}_{\textbf{2}}{\vert}{\textbf{X}}_{\textbf{1}},{\textbf{X}}_{\textbf{2}},{\textbf{Y}}_{\textbf{1}})-h({\textbf{V}}_{\textbf{1}},{\textbf{V}}_{\textbf{2}})+h(\textbf{V}_{\textbf{1}}{\vert}\textbf{U}_{\textbf{1}})+h(\textbf{V}_{\textbf{2}}{\vert}\textbf{U}_{\textbf{1}},\textbf{U}_{\textbf{2}},\textbf{V}_{\textbf{1}}) \notag\\
&\stackrel{(a)}=-h(\textbf{N}_{\textbf{1}})-h(\textbf{N}_{\textbf{2}}{\vert}\textbf{X}_{\textbf{1}},\textbf{X}_{\textbf{2}},\textbf{Y}_{\textbf{1}})+h(\textbf{W}_{\textbf{1}})+h(\textbf{W}_{\textbf{2}}{\vert}\textbf{U}_{\textbf{1}},\textbf{U}_{\textbf{2}},\textbf{V}_{\textbf{1}}) \notag\\
&\stackrel{(b)}=-h(\textbf{N}_{\textbf{2}})+h(\textbf{W}_{\textbf{2}}) \stackrel{(c)}=0,
\end{align}
where (a), (b) and (c) follow since, given $\textbf{X}_{\textbf{2}}$, the uncertainty in $\textbf{Y}_{\textbf{2}}$ is only due to $\textbf{N}_{\textbf{2}}$ and, for Gaussian noise channels, $\textbf{W}_{\textbf{1}}$ and $\textbf{W}_{\textbf{2}}$ are i.i.d. with the same distribution as $\textbf{N}_{\textbf{1}}$.
\end{proof}
\noindent Now define
\begin{align}
V_{\lambda}^{q}({\textbf{K}^{'}})&= \sup_{\textbf{X}: \e{{\textbf{X}}{\textbf{X}^{\textbf{T}}}} {\preccurlyeq} {\textbf{K}^{'}}} \overrightarrow{S}_{\lambda}^{q}(\textbf{X}) \notag\\
&= \sup_{\substack{(V,\textbf{X}): \e{{\textbf{X}}{\textbf{X}^{\textbf{T}}}} {\preccurlyeq} {\textbf{K}^{'}}\\ V \to \textbf{X} \to (\textbf{Y},\textbf{Z})}} \overrightarrow{s}_{\lambda}^{q}({\textbf{X}}{\vert}V). \label{eqnmax}
\end{align}
By Proposition $7$ of \cite{geng2012capacity}, there exists a pair of random variables $(V^{\dagger},\textbf{X}^{\dagger})$ with $\left\vert{{\mathcal{V}}^{\dagger}}\right\vert \leq 1+\frac{d(d+1)}{2}$ and \smash{$\e{{\textbf{X}}^{\dagger}{{\textbf{X}^{\dagger}}^{\textbf{T}}}} {\preccurlyeq} {\textbf{K}^{'}}$} such that (with $d$ being the dimension of all the vectors)
\begin{equation}
V_{\lambda}^{q}({\textbf{K}^{'}})= s_{\lambda}^{q}(\textbf{X}^{\dagger}{\vert}V^{\dagger})=\overrightarrow{s}_{\lambda}^{q}(\textbf{X}^{\dagger}{\vert}V^{\dagger}).
\end{equation}
\noindent The following lemma helps in identifying the optimal distribution that achieves $V_{\lambda}^{q}({\textbf{K}^{'}})$.
\begin{lemma} \label{lem:final}
Let $(V^{\dagger},\textbf{X}^{\dagger}) \sim p^{\dagger}(v,\textbf{x})$ attain $V_{\lambda}^{q}({\textbf{K}^{'}})$. Let $\textbf{X}_{\textbf{v}}$ be distributed according to the conditional law $p^{\dagger}(\textbf{X}{\vert}V=v)$. Let $(V_{1},V_{2},\textbf{X}_{\textbf{1}},\textbf{X}_{\textbf{2}}) \sim p(v_{1},v_{2},\textbf{x}_{\textbf{1}},\textbf{x}_{\textbf{2}})$ be such that the marginals  $p^{\dagger}(v_{1},\textbf{x}_{\textbf{1}})$ and $p^{\dagger}(v_{2},\textbf{x}_{\textbf{2}})$ attain $V_{\lambda}^{q}({\textbf{K}^{'}})$. Define
\begin{gather}
\textbf{U}_{\textbf{1}}{\vert}((V_{1},V_{2})=(v_{1},v_{2})) \sim \frac{1}{\sqrt{2}}(\textbf{X}_{\textbf{v}_{\textbf{1}}}+\textbf{X}_{\textbf{v}_{\textbf{2}}})\\
\textbf{U}_{\textbf{2}}{\vert}((V_{1},V_{2})=(v_{1},v_{2})) \sim \frac{1}{\sqrt{2}}(\textbf{X}_{\textbf{v}_{\textbf{1}}}-\textbf{X}_{\textbf{v}_{\textbf{2}}}).
\end{gather} 
Then $\textbf{U}_{\textbf{1}}$ and $\textbf{U}_{\textbf{2}}$ attain $V_{\lambda}^{q}({\textbf{K}^{'}})$.
\end{lemma}
\begin{proof}
Consider the two letter broadcast channel $q(\textbf{y}_{\textbf{1}},\textbf{z}_{\textbf{1}}{\vert}\textbf{x}_{\textbf{1}}){\times}q(\textbf{y}_{\textbf{2}},\textbf{z}_{\textbf{2}}{\vert}\textbf{x}_{\textbf{2}})$. Let $\textbf{Y}^a_{\textbf{1}}{\vert}((V_{1},V_{2})=(v_{1},v_{2})) \sim \frac{1}{\sqrt{2}}({\textbf{Y}}_{v_{1}}+{\textbf{Y}}_{v_{2}})$, $\textbf{Y}^a_{\textbf{2}}{\vert}((V_{1},V_{2})=(v_{1},v_{2})) \sim \frac{1}{\sqrt{2}}({\textbf{Y}}_{v_{1}}-{\textbf{Y}}_{v_{2}})$, $\textbf{Z}^a_{\textbf{1}}{\vert}((V_{1},V_{2})=(v_{1},v_{2})) \sim \frac{1}{\sqrt{2}}({\textbf{Z}}_{v_{1}}+{\textbf{Z}}_{v_{2}})$ and $\textbf{Z}^a_{\textbf{2}}{\vert}((V_{1},V_{2})=(v_{1},v_{2})) \sim \frac{1}{\sqrt{2}}({\textbf{Z}}_{v_{1}}-{\textbf{Z}}_{v_{2}})$. We have
\begin{align}
2V_{\lambda}^{q}({\textbf{K}}^{'}) &{\leq} \max_{\substack{p(v_{1},v_{2}{\vert}{\textbf{x}}_{\textbf{1}},{\textbf{x}}_{\textbf{2}})\\ \e{{\textbf{X}}_{\textbf{1}}{\textbf{X}}_{\textbf{1}}^{\textbf{T}}} {\preccurlyeq} {\textbf{K}^{'}}, \e{{\textbf{X}}_{\textbf{2}}{{\textbf{X}}_{\textbf{2}}^{\textbf{T}}}} {\preccurlyeq} {\textbf{K}^{'}}\\ (V_{1},V_{2}) \to ({\textbf{X}}_{\textbf{1}},{\textbf{X}}_{\textbf{2}}) \to ({\textbf{Y}}_{\textbf{1}},{\textbf{Y}}_{\textbf{2}})}} \overrightarrow{s}_{\lambda}^{q{\times}q}({\textbf{X}}_{\textbf{1}},{\textbf{X}}_{\textbf{2}}{\vert}V_{1},V_{2}) \notag\\
&\stackrel{(a)} = \max_{\substack{p(v_{1},v_{2}{\vert}{\textbf{u}}_{\textbf{1}},{\textbf{u}}_{\textbf{2}})\\ \e{{\textbf{U}}_{\textbf{1}}{{\textbf{U}}_{\textbf{1}}^{\textbf{T}}}} {\preccurlyeq} {{\textbf{K}}^{'}}, \e{{\textbf{U}}_{\textbf{2}}{{\textbf{U}}_{\textbf{2}}^{\textbf{T}}}} {\preccurlyeq} {\textbf{K}^{'}}\\ (V_{1},V_{2}) \to ({\textbf{U}}_{\textbf{1}},{\textbf{U}}_{\textbf{2}}) \to ({\textbf{Y}}^a_{\textbf{1}},{\textbf{Y}}^a_{\textbf{2}})}} \overrightarrow{s}_{\lambda}^{q{\times}q}({\textbf{U}}_{\textbf{1}},{\textbf{U}}_{\textbf{2}}{\vert}V_{1},V_{2}) \notag\\
&\stackrel{(b)} = \overrightarrow{S}_{\lambda}^{q{\times}q}({\textbf{U}}_{\textbf{1}},{\textbf{U}}_{\textbf{2}}) \notag\\
&\stackrel{(c)} {\leq} \overrightarrow{S}_{\lambda}^{q}({\textbf{U}}_{\textbf{1}})+\overrightarrow{S}_{\lambda}^{q}({\textbf{U}}_{\textbf{2}}) \stackrel{(d)} {\leq} 2V_{\lambda}^{q}({\textbf{K}}^{{'}}),
\end{align}
where (a) follows from Property \ref{lem:rot} and 
\begin{align}
&\e{{\textbf{U}_{\textbf{1}}^T}{\textbf{U}_{\textbf{1}}}}=\e{{\textbf{U}_{\textbf{2}}^T}{\textbf{U}_{\textbf{2}}}} \notag\\
&=\sum_{v_{1},v_{2}} p_{\ast}(v_{1},v_{2})\frac{1}{2}(\e{{\textbf{X}_{\textbf{v}_1}^T}{\textbf{X}_{\textbf{v}_1}}}+\e{{\textbf{X}_{\textbf{v}_2}^T}{\textbf{X}_{\textbf{v}_2}}}) {\preccurlyeq} {\textbf{K}^{'}},
\end{align}
(b) follows by the definition of concave envelope, (c) follows from Property \ref{lem:2letfact} and (d) follows by definition of $V_{\lambda}^{q}({\textbf{K}^{'}})$. Since the extremes match, all inequalities are equalities. In particular, (d) is an equality. Thus $\textbf{U}_{\textbf{1}}$ and $\textbf{U}_{\textbf{2}}$ attain $V_{\lambda}^{q}({\textbf{K}^{'}})$.
\end{proof}
\noindent The fact in Lemma \ref{lem:final} allows us to identify that Gaussian auxiliaries are indeed optimal. This is stated in the form of the following theorem.
\begin{theorem} \label{thm:final}
There exists $\textbf{X}^{\dagger} \sim \mathcal{N}(\textbf{0},\textbf{K}^{\dagger})$, $\textbf{K}^{\dagger} {\preccurlyeq} {\textbf{K}^{'}}$ such that $V_{\lambda}^{q}({\textbf{K}^{'}})=\overrightarrow{s}_{\lambda}^{q}(\textbf{X}^{\dagger})$.
\end{theorem}
\begin{proof}
From Lemma \ref{lem:final}, ${\textbf{X}}_{{v_{1}}}+{\textbf{X}}_{{v_{2}}}$ attains $V_{\lambda}^{q}({\textbf{K}^{'}})$. Thus for any optimal distribution $p_{o}$, the distribution $p_{o} * p_{o}$ works equally well. Continuing this argument, by induction, we can use the Central Limit Theorem to conclude that Gaussian is a maximizer. Let $\textbf{X}^{\dagger} \sim \mathcal{N}(\textbf{0},\textbf{K}^{\dagger})$ be the maximizer, where $\textbf{K}^{\dagger} {\preccurlyeq} {\textbf{K}^{'}}$. Note that since $v_1$ and $v_2$ were arbitrary, all the $\textbf{X}_v$'s have the same covariance matrix, namely $\textbf{K}^{\dagger}$.
\end{proof}
\noindent We now complete the proof of Theorem \ref{thm:pdvecfb} using Theorem \ref{thm:final}.

\subsection{Proof of Theorem \ref{thm:pdvecfb}}
In expression \eqref{equi}, the first term on the LHS is maximized by $\textbf{X} \sim \mathcal{N}(\textbf{0},{\textbf{K}^{'}})$. The second term is maximized by $\textbf{U}={\textbf{X}}^{\dagger} \sim \mathcal{N}(\textbf{0},\textbf{K}^{\dagger})$, where $\textbf{K}^{\dagger} {\preccurlyeq} {\textbf{K}^{'}}$ by Theorem \ref{thm:final}. Now let $\textbf{X}={\textbf{X}}^{\dagger}+{\textbf{X}}^{\ast}=\textbf{U}+\textbf{V}$, where $\textbf{U}={\textbf{X}}^{\dagger} \sim \mathcal{N}(\textbf{0},{\textbf{K}}^{\dagger})$ and $\textbf{V}={\textbf{X}}^{\ast} \sim \mathcal{N}(\textbf{0},({\textbf{K}^{'}}-{\textbf{K}}^{\dagger}))$ are independent. This gives $\textbf{X} \sim \mathcal{N}(\textbf{0},{\textbf{K}^{'}})$ as desired. Evaluating  $I(\textbf{X};\textbf{Z})$ and $\overrightarrow{s}_{\lambda}^{q}(\textbf{X}{\vert}V)$ with these choices, it can be easily checked that \eqref{equi} is indeed satisfied. Hence we have proved the reverse inclusion ${\mathcal{O}}_{{\textbf{K}^{
 '}}} \subseteq {\mathcal{I}}_{{\textbf{K}^{'}}}$. This establishes that the capacity region of a PD-GVBC is not enlarged by feedback. We have also proved that a superposition of Gaussian codes is optimal.
\begin{remark}
Although our proofs were given for the model 
$\mathbf{Y}=\mathbf{X}+\mathbf{N}$ and $\mathbf{Z}=\mathbf{Y}+\widetilde{\mathbf{N}}$, it can be easily verified that the proofs go through for a more general model $\mathbf{Y}=G \mathbf{X}+\mathbf{N}$ and $\mathbf{Z}=\mathbf{Y}+\widetilde{\mathbf{N}}$ (where $G$ is an invertible matrix) as well. This justifies our terminology of Gaussian Vector BC. The capacity region in this more general case is given in the form of the following theorem:
\begin{theorem} \label{thm:gen}
The capacity region $\mathcal{C}_{MVBC}^{fb}$ of a MIMO PD-GBC with perfect feedback from both the receivers to the transmitter is the union of the set of $(R_{1},R_{2})$ pairs such that
\begin{gather}
R_{1} \leq \frac{1}{2} \log\Bigg(\frac{\left\vert{G^{T}\textbf{B}_{\textbf{1}}G+\textbf{K}}\right\vert}{\left\vert{\textbf{K}}\right\vert}\Bigg) \\
R_{2} \leq \frac{1}{2} \log\Bigg(\frac{\left\vert{ G^{T}(\textbf{B}_{\textbf{2}}+\textbf{B}_{\textbf{1}})G+\textbf{K}+\tilde{\textbf{K}}}\right\vert} {\left\vert{ G^{T}\textbf{B}_{\textbf{1}}G+\textbf{K}+\tilde{\textbf{K}}}\right\vert}\Bigg),
\end{gather}
for some $\textbf{B}_{\textbf{1}},\textbf{B}_{\textbf{2}}\succcurlyeq0$, with $\textbf{B}_{\textbf{1}}+\textbf{B}_{\textbf{2}}\:{\preccurlyeq}\:\textbf{K}^{'}$.
\end{theorem}
\end{remark}

\section{Feedback Capacity of Reversely PD DMBC}\label{rddmbc}
 We now extend our results to a reversely physically degraded DM broadcast channel (RPDBC). The model is as 
shown in Fig.~2. As seen in the figure, the terminology RPDBC is used to refer to a product of two (inconsistently) degraded DM broadcast channels. El Gamal~\cite{gamal1980pacity} derived the capacity region of a RPDBC with 
both private and common messages in the absence of feedback. In general, while the effect of feedback on a RPDBC has not been 
studied so far in literature, we show that feedback does not enlarge the capacity region of a RPDBC. 
This enables our results from the previous sections to be adapted so as to establish the capacity region 
of a Gaussian reversely PD vector BC with feedback. Here, we focus on deriving the discrete memoryless 
result. The definition of a code, achievable rate and capacity region are similar to \cite{gamal1980pacity}, 
except that we take 
\begin{align}
X_{ji} &= g_j(W_0,W_1,W_2,Y_1^{i-1},Z_1^{i-1},Y_2^{i-1},Z_2^{i-1}), \, j = 1,2, 
\end{align}
for some maps $g_1(\cdot)$ and $g_2(\cdot)$. Let us denote the capacity region of a RPDBC with both private and common messages with perfect noiseless feedback from all the outputs by $\mathcal{C}_{RBC}^{fb}$. 
\begin{figure}[h]
\centering
\includegraphics[scale=1.2]{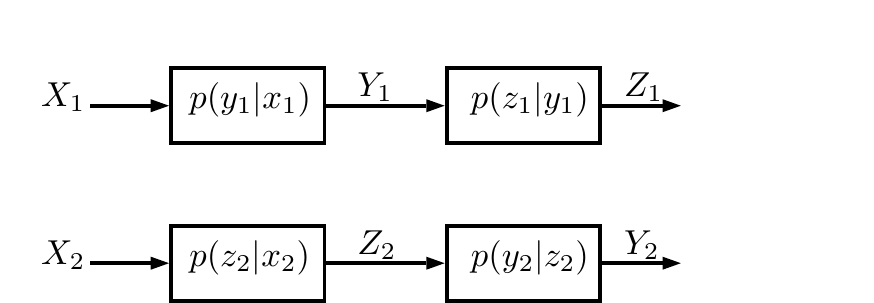}
\label{rdbc1}
\caption{Reversely Physically Degraded BC}
\end{figure}
\begin{theorem}\label{thm:rbc}
\begin{equation}
\mathcal{C}_{RBC}^{fb} = \bigcup \{(R_0,R_1,R_2)\},
\end{equation}
where
\begin{align}
R_0 &\leq \min \Bigl\{\sum_{i=1}^2 I(U_i;Y_i), \sum_{i=1}^2 I(U_i;Z_i) \Bigr\} \notag\\
R_0+R_1 &\leq I(X_1;Y_1)+I(U_2;Y_2) \notag\\
R_0+R_2 &\leq I(X_2;Z_2)+I(U_1;Z_1) \notag\\
R_0+R_1+R_2 &\leq I(X_1;Y_1)+I(U_2;Y_2)+I(X_2;Z_2|U_2) \label{rdc}\\
R_0+R_1+R_2 &\leq I(X_2;Z_2)+I(U_1;Z_1)+I(X_1;Y_1|U_1), \notag
\end{align}
with the union taken over all $p(u_1,x_1)p(u_2,x_2)$ such that $U_1 \to X_1 \to Y_1 \to Z_1$ and $U_2 \to X_2 \to Z_2 \to Y_2$, and $|\mathcal{U}_i| \leq \min (|\mathcal{X}_i|,|\mathcal{Y}_i|,|\mathcal{Z}_i|), i=1,2$. 
\end{theorem}
\noindent Note that the region mentioned above is exactly the region without feedback, which was established in \cite{gamal1980pacity}.
\begin{proof}
The achievability follows from \cite{gamal1980pacity} by ignoring the feedback. Since the steps in the converse proof are similar to \cite{gamal1980pacity}, we only highlight the differences. It turns out that employing the following auxiliaries
\begin{align}
U_{1i} &= (W_0,W_2,Y_1^{i-1},Z_1^{i-1},Y_2^n)\\
U_{2i} &= (W_0,W_1,Z_2^{i-1},Y_2^{i-1},Z_1^n)
\end{align}
in the converse proof of \cite{gamal1980pacity} suffices to obtain the theorem. Observe that the 
above choice of auxiliaries ensures both the Markov conditions $U_{1i} \to X_{1i} \to Y_{1i} \to Z_{1i}$ 
and $U_{2i} \to X_{2i} \to Z_{2i} \to Y_{2i}$. These conditions are not true for the original auxiliaries in 
\cite{gamal1980pacity}.
%

We show the single-letterization of  \eqref{rdc}, others follow in a similar fashion. 
Also, we suppress the $n\epsilon_n$ terms while applying Fano's inequality in the expression below.
\begin{align}
&n(R_0+R_1+R_2) = H(W_0,W_1)+H(W_2|W_0,W_1) \notag\\
&\leq I(W_0,W_1;Y_1^n,Y_2^n)+I(W_2;Z_1^n,Z_2^n|W_0,W_1) \notag\\
&= I(W_0,W_1;Y_2^n)+I(W_0,W_1;Y_1^n|Y_2^n)+I(W_2;Z_2^n|W_0,W_1,Z_1^n)+I(W_2;Z_1^n|W_0,W_1) \notag\\
&= I(W_0,W_1,Z_1^n;Y_2^n)-I(Z_1^n;Y_2^n|W_0,W_1) +I(W_2;Z_2^n|W_0,W_1,Z_1^n)+I(W_0,W_1;Y_1^n|Y_2^n) +I(W_2;Z_1^n|W_0,W_1) \notag\\
&= I(W_0,W_1,Z_1^n;Y_2^n)-I(Z_1^n;Y_2^n|W_0,W_1) +I(W_2;Z_2^n|W_0,W_1,Z_1^n)+I(W_0,W_1;Y_1^n|Y_2^n) \notag\\
&\phantom{ww}+I(W_2;Z_1^n,Y_2^n|W_0,W_1)-I(W_2;Y_2^n|W_0,W_1,Z_1^n) \notag\\
&= I(W_0,W_1,Z_1^n;Y_2^n)-I(Z_1^n;Y_2^n|W_0,W_1)+I(W_2;Z_2^n|W_0,W_1,Z_1^n)+I(W_0,W_1;Y_1^n|Y_2^n) \notag\\
&\phantom{ww}+I(W_2;Z_1^n|W_0,W_1,Y_2^n)+I(W_2;Y_2^n|W_0,W_1)-I(W_2;Y_2^n|W_0,W_1,Z_1^n) \notag\\
&= I(W_0,W_1,Z_1^n;Y_2^n)-I(Z_1^n;Y_2^n|W_0,W_1) +I(W_2;Z_2^n|W_0,W_1,Z_1^n)+I(W_0,W_1;Y_1^n|Y_2^n) \notag\\
&\phantom{ww}+I(W_2;Z_1^n|W_0,W_1,Y_2^n)+H(Y_2^n|W_0,W_1) -H(Y_2^n|W_0,W_1,W_2)-H(Y_2^n|W_0,W_1,Z_1^n) \notag\\
&\phantom{ww}+H(Y_2^n|W_0,W_1,Z_1^n,W_2) \notag\\
&= I(W_0,W_1,Z_1^n;Y_2^n)+I(W_2;Z_2^n|W_0,W_1,Z_1^n) +I(W_0,W_1;Y_1^n|Y_2^n)+I(W_2;Z_1^n|W_0,W_1,Y_2^n) \notag\\
&\phantom{ww}-I(Y_2^n;Z_1^n|W_0,W_1,W_2) \notag\\
&\leq I(W_0,W_1,Z_1^n;Y_2^n)+I(W_2;Z_2^n|W_0,W_1,Z_1^n) +I(W_0,W_1;Y_1^n|Y_2^n)+I(W_2;Z_1^n|W_0,W_1,Y_2^n) \notag\\
&\leq I(W_0,W_1,Z_1^n;Y_2^n)+I(W_2;Z_2^n|W_0,W_1,Z_1^n) +I(W_0,W_1;Y_1^n|Y_2^n)+I(W_2;Y_1^n|W_0,W_1,Y_2^n) \notag\\
&\leq I(W_0,W_1,Z_1^n;Y_2^n)+I(W_2;Z_2^n|W_0,W_1,Z_1^n) +H(Y_1^n|Y_2^n)-H(Y_1^n|W_0,W_1,W_2,Y_2^n) \notag\\
&\leq I(W_0,W_1,Z_1^n;Y_2^n)+I(W_2;Z_2^n|W_0,W_1,Z_1^n) +H(Y_1^n)-H(Y_1^n|W_0,W_1,W_2,Y_2^n,X_1^n) \notag\\
&= I(W_0,W_1,Z_1^n;Y_2^n)+I(W_2;Z_2^n|W_0,W_1,Z_1^n) +H(Y_1^n)-H(Y_1^n|X_1^n) \notag\\
&= I(W_0,W_1,Z_1^n;Y_2^n)+I(X_1^n;Y_1^n) +I(W_2;Z_2^n|W_0,W_1,Z_1^n) \notag\\
&= \sum_{i=1}^n \{I(W_0,W_1,Z_1^n;Y_{2i}|Y_2^{i-1})+H(Y_{1i}|Y_1^{i-1}) +I(W_2;Z_{2i}|Z_2^{i-1},W_0,W_1,Z_1^n)-H(Y_{1i}|Y_1^{i-1},X_1^n)\} \notag\\
&\leq \sum_{i=1}^n \{I(W_0,W_1,Z_1^n,Y_2^{i-1};Y_{2i})-H(Y_{1i}|X_{1i}) +H(Y_{1i})+I(W_2;Z_{2i}|Z_2^{i-1},W_0,W_1,Z_1^n,Y_2^{i-1})\} \notag\\
&\leq \sum_{i=1}^n \{I(W_0,W_1,Z_1^n,Y_2^{i-1},Z_2^{i-1};Y_{2i})+I(X_{1i};Y_{1i}) +I(W_2;Z_{2i}|Z_2^{i-1},W_0,W_1,Z_1^n,Y_2^{i-1})\} \notag\\
&= \sum_{i=1}^n \{I(U_{2i};Y_{2i})+I(W_2;Z_{2i}|U_{2i})+I(X_{1i};Y_{1i})\} \notag\\
&= \sum_{i=1}^n \{I(X_{1i};Y_{1i})+I(U_{2i};Y_{2i})+H(Z_{2i}|U_{2i}) -H(Z_{2i}|W_0,W_1,W_2,Z_1^n,Y_2^{i-1},Z_2^{i-1})\} \notag\\
&\leq \sum_{i=1}^n \{I(X_{1i};Y_{1i})+I(U_{2i};Y_{2i})+H(Z_{2i}|U_{2i}) -H(Z_{2i}|W_0,W_1,W_2,Z_1^n,Y_2^{i-1},Z_2^{i-1},Y_1^{i-1},X_{2i})\} \notag\\
&= \sum_{i=1}^n \{I(X_{1i};Y_{1i})+I(U_{2i};Y_{2i})+H(Z_{2i}|U_{2i}) -H(Z_{2i}|X_{2i})\} \notag\\
&\leq \sum_{i=1}^n \{I(X_{1i};Y_{1i})+I(U_{2i};Y_{2i})+H(Z_{2i}|U_{2i})  - H(Z_{2i}|X_{2i},U_{2i})\} \notag\\
&= \sum_{i=1}^n \{I(X_{1i};Y_{1i})+I(U_{2i};Y_{2i})+I(X_{2i};Z_{2i}|U_{2i})\}
\end{align}
The single letterization is completed by letting $n \to \infty$ which makes $\epsilon_n \to 0$ and noting that the region in Theorem \ref{thm:rbc} is convex.
\end{proof}

\section{Conclusion}\label{concl}
We proved that feedback does not enlarge the capacity region of a physically degraded Gaussian vector BC. Our proof does not rely on the use of the EPI or any of its variants, and makes use of factorization of concave envelopes in terms of directed information. It illustrates the utility of the technique of factorization of concave envelopes as a tool for proving converses in Gaussian settings, where previously the application of EPI was the main bottleneck. Our work also brings out the fact that factorizing concave envelopes of directed information can handle situations involving feedback, where concave envelopes of mutual information are not factorizable. We also proved that the capacity region of a reversely physically degraded DM broadcast channel is not enlarged by feedback.

\section*{Acknowledgements}
The work was supported in part by the Bharti Centre for Communication, IIT Bombay and in part by the Department of Science and Technology (DST), Government of India, under grant SB/S3/EECE/077/2013.

\bibliographystyle{IEEEtran}
\bibliography{refpdgbc}

\begin{thebibliography}{1}
\providecommand{\url}[1]{#1}
\csname url@samestyle\endcsname
\providecommand{\newblock}{\relax}
\providecommand{\bibinfo}[2]{#2}
\providecommand{\BIBentrySTDinterwordspacing}{\spaceskip=0pt\relax}
\providecommand{\BIBentryALTinterwordstretchfactor}{4}
\providecommand{\BIBentryALTinterwordspacing}{\spaceskip=\fontdimen2\font plus
\BIBentryALTinterwordstretchfactor\fontdimen3\font minus
  \fontdimen4\font\relax}
\providecommand{\BIBforeignlanguage}[2]{{%
\expandafter\ifx\csname l@#1\endcsname\relax
\typeout{** WARNING: IEEEtran.bst: No hyphenation pattern has been}%
\typeout{** loaded for the language `#1'. Using the pattern for}%
\typeout{** the default language instead.}%
\else
\language=\csname l@#1\endcsname
\fi
#2}}
\providecommand{\BIBdecl}{\relax}
\BIBdecl

\bibitem{gamal1978feedback}
A.~E. Gamal, ``The feedback capacity of degraded broadcast channels
  (corresp.),'' \emph{Information Theory, IEEE Transactions on}, vol.~24,
  no.~3, pp. 379--381, 1978.

\bibitem{gah4al1981capacity}
------, ``The capacity of the physically degraded gaussian broadcast channel
  with feedback,'' \emph{Information Theory, IEEE Transactions on}, vol.~21,
  no.~4, 1981.

\bibitem{el2011network}
A.~El~Gamal and Y.-H. Kim, \emph{Network information theory}.\hskip 1em plus
  0.5em minus 0.4em\relax Cambridge University Press, 2011.

\bibitem{weingarten2006capacity}
H.~Weingarten, Y.~Steinberg, and S.~Shamai, ``The capacity region of the
  gaussian multiple-input multiple-output broadcast channel,''
  \emph{Information Theory, IEEE Transactions on}, vol.~52, no.~9, pp.
  3936--3964, 2006.

\bibitem{geng2012capacity}
Y.~Geng and C.~Nair, ``The capacity region of the two-receiver gaussian vector
  broadcast channel with private and common messages,'' \emph{Information
  Theory, IEEE Transactions on}, vol.~60, no.~4, pp. 2087--2104, April 2014.

\bibitem{massey1990causality}
J.~Massey, ``Causality, feedback and directed information,'' in \emph{Proc.
  Int. Symp. Inf. Theory Applic.(ISITA-90)}.\hskip 1em plus 0.5em minus
  0.4em\relax Citeseer, 1990, pp. 303--305.

\bibitem{gamal1980pacity}
A.~E. Gamal, ``Capacity of the sum and product of two unmatched broadcast
  channels,'' \emph{Problemy Peredachi Informatsii}, vol.~16, pp. 3--23, 1980.

\end{thebibliography}

\end{document}